\documentclass[conference,10pt]{IEEEtran}
\IEEEoverridecommandlockouts
\makeatletter
\def\ps@headings{%
\def\@oddhead{\mbox{}\scriptsize\rightmark \hfil \thepage}%
\def\@evenhead{\scriptsize\thepage \hfil \leftmark\mbox{}}%
\def\@oddfoot{}%
\def\@evenfoot{}}
\makeatother

\usepackage{amssymb,amsmath}
\usepackage{paralist}
\usepackage{graphicx}
\usepackage{epsfig}
\usepackage{pst-all,pstricks-add}
\usepackage{algorithm2e}

\usepackage{color}
\usepackage{arydshln}
\usepackage{cite,manfnt}
\usepackage{tikz}

\newcommand{\F}{\mathbf{F}}

\newcommand{\N}{\mathcal{N}}

\newtheorem{theorem}{\textbf{Theorem}}
\newtheorem{lemma}[theorem]{\textbf{Lemma}}
\newtheorem{definition}[theorem]{\textbf{Definition}}

\newcommand{\nix}[1]{}

\begin{document}
\title{\LARGE{\textbf{Network  Protection Design Using Network Coding}}}
\author{
\authorblockN{Salah A. Aly}
\authorblockA{Department of EE\\ Princeton University\\ salah@princeton.edu}
\and
\authorblockN{Ahmed E. Kamal}
\authorblockA{Department of ECE\\ Iowa State University\\ kamal@iastate.edu}
\and
\authorblockN{Anwar I. Walid}
\authorblockA{Bell Labs, Alcatel-Lucent\\ Murry Hill, NJ \\ anwar@research.bell-labs.com}
 }
\maketitle

\begin{abstract}
Link and node failures  are two common  fundamental problems that affect operational networks.
Protection of communication networks against such failures is essential for maintaining network reliability and performance.
Network protection codes (NPC) are proposed to protect operational networks against link and node failures. Furthermore, encoding and decoding operations of such codes are well
developed over binary and finite fields. Finding network topologies, practical scenarios, and limits on graphs applicable for NPC are of interest. In this paper, we establish limits on network protection design.   We
investigate several network graphs where NPC can be deployed using network coding. Furthermore, we construct graphs with minimum number of edges suitable for network protection codes deployment.
\end{abstract}

\section{Introduction}
With the increase in the capacity of backbone networks, the failure of
a single link or node can result in the loss of a significant amount of
information, which may lead to loss of revenues or even catastrophic failures. Network connections are therefore provisioned with the property that they can survive such edge and node failures. Several techniques have been introduced in the literature to achieve such goal, where either extra resources are added or some of the
available network resources are reserved as backup circuits.
Recovery from failures is also required to be agile in order to
minimize the network outage time.
This recovery usually involves two steps: fault diagnosis  and connections rerouting .
Hence, the optimal network survivability problem is
a multi-objective problem in terms of resource efficiency, operation cost, and agility~\cite{zeng07,zhang06}.

Allowing network relay nodes to encode packets is a novel approach that has attracted much research work from both academia and industry with applications in enterprise networks, wireless communication and storage systems. This approach, which is known as network coding, offers benefits such as minimizing network delay, maximizing
network capacity and enabling security and protection services, see~\cite{soljanin07,yeung06} and references therein. Network coding allows the
sender nodes to combine/encode the incoming packets into one outgoing packet. Furthermore, the receiver nodes are allowed to decode
those packets once they receive enough number of combinations. However,
finding practical network topologies where network coding can be deployed is a
challenging problem.
In order to apply network coding on a network with a large number of
nodes, one must ensure that the encoding and decoding operations are done correctly over binary and finite fields.

There have been several applications for the edge disjoint paths (EDP) and node disjoint paths (NDP)
problems in the literature including network flow, traffic routing, load balancing and optimal network design.  In both cases (edge and vertex disjointness paths), deciding whether
the pairs can be disjointedly connected is NP-complete~\cite{vygen95}.

A network protection scheme against a single link failure using network coding and reduced capacity is shown in~\cite{aly08preprint1}. The scheme is extended to protect against multiple link failures as well as against a single node failure. A protection scheme protects the communication links and network traffic between a group of senders and receivers in a large network with several relay nodes. This scheme is based on what we  call \emph{Network Protection Codes} (NPCs), which are defined in Section~\ref{sec:terminology}. The encoding and decoding operations of such codes are defined in the case of binary and finite fields in~\cite{aly08preprint1,aly08i}. In
this paper, we establish limits on network protection codes  and
investigate several network graphs where NPC can be deployed. In addition, we construct graphs with minimum number of edges to facilitate NPC deployment.


This paper is organized as follows.
In Section~\ref{sec:terminology}  we present  the network model and essential definitions.  In Section~\ref{sec:NPC-minimumedges}, we derive bounds on the minimum number of edges of graphs for NPC, and construct graphs that meet these bounds in Section~\ref{sec:graphconstruction}. Section~\ref{sec:k-connectedgraph} presents limits on certain graphs that are applicable for NPC deployment.

\section{Network Model and NPC Definition}\label{sec:terminology}
In this section we present the network model, define briefly network protection codes, and then state the problem. Further details can be found in~\cite{aly08preprint1}.

\subsection{Network Model}
The network model is described as follows:
\begin{compactenum}[i)]
\item Let $\N$ be a network represented by an abstract graph
    $G=(\textbf{V},E)$, where $\textbf{V}$ is the set of nodes and $E$
    is   set of undirected edges. Let $S$ and $R$ be   sets of independent sources
    and destinations, respectively. The set $\textbf{V}=V\cup S \cup
    R$ contains the relay, source, and destination nodes, respectively.

\item The node can be a router, switch, or an end terminal depending on
    the network model $\N$ and the transmission layer.

\item Let $L$ be a set of links $L=\{L_1,L_2,\ldots,L_k\}$ carrying the data
    from the sources to the receivers.
    All connections have the same bandwidth, otherwise a connection with
    high bandwidth can be divided into multiple connections, each of which
    has a unit capacity. There are exactly $k$ connections. For
    simplicity, we assume that the number of sources is less than or equal
    to the number of links. A sender with a high capacity can divide its
    capacity into multiple unit capacities, each of which has its own link. In other words, \begin{eqnarray}L_i=
    \{(s_i,w_{1i}),(w_{1i},w_{2i}),\ldots,(w_{(\lambda)i},r_i) \},\end{eqnarray} where
    $1\leq i\leq k$ and $(w_{(j-1)i},w_{ji}) \in E$, for some integer $\lambda \geq 1$.
\item The failure on a link $L_i$ may occur due to  network
    circumstances such as a link replacement and overhead. We
    assume that the receiver is able to detect a failure and is able to use a    protection strategy  to recover it. We will use in the rest of the paper the terms edges and links interchangeably
\end{compactenum}
\begin{figure}[t!]
 \begin{center} 
  \includegraphics[height=4.8cm,width=7.5cm]{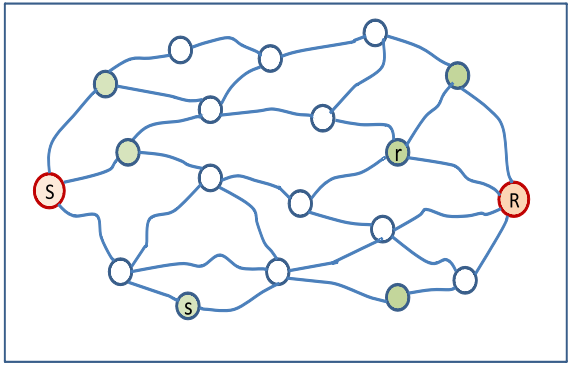}
  \caption{A network model with a super source S and a super receiver R. A set of sources and a set of receivers are shown in green color. We assume the source(s) and receiver(s) share $k$ edge disjoint paths.}\label{fig:npaths}
\end{center}
\end{figure}
\subsection{NPC Definition}
Let us assume a network model $\N$ with $t>1$ path failures in the $k$ working paths, i.e., paths carrying data from source(s) to receiver(s)~\cite{aly08preprint1}.  One can
define a\emph{ network protection code} NPC which protects $k$ edge disjoint paths as shown in
the systematic matrix $G$ in Eq.~(\ref{eq:Gmultiple}). In general, the  systematic matrix $G$
defines the source nodes that will send encoded  messages and source
nodes that will send only plain message without encoding. In order to protect $k$ working paths, $k-t$ connections must carry plain data, and $t$ connections must carry encoded data.


The generator matrix of the NPC for multiple link failures is given by:


\begin{eqnarray}\label{eq:Gmultiple}
G\!\!=\!\! \left[\!\!\!
\begin{array}{c|c}
\begin{array}{cccc}1&0&\ldots&0\!\!\!\\
0& 1&\ldots&0\!\!\!\\
\vdots&\vdots&\vdots&\!\!\!\\
 0&0&\ldots&1\!\!\!\\
 \end{array}\!\!\!&\begin{array}{ccccc}p_{11}\!&\!p_{12}&\ldots&p_{1t}\\
p_{21}\!&\!p_{22}&\ldots&p_{2t}\\
\vdots&\vdots&\vdots&\vdots\\ \!\!\!
p_{k-t,1}&p_{k-t,2}&\ldots&p_{k-t,t}\!\!\\
 \end{array} \\\\
 \multicolumn{2}{c}{}\\ \underbrace{\hskip 0.7in}^{\mbox{ident. $I_{k-t\times~ k-t}$}} &\underbrace{\hspace{0.7in}}^{\mbox{Submatrix } P_{k-t \times t}}\\
\end{array}\!\!\!
  \right]\!,\!\!\!\!
\end{eqnarray}
where $p_{ij} \in \F_q$, and $q\geq k-t+1$, see~\cite{aly08preprint1}.

The matrix $G$ can be rewritten as
\begin{eqnarray}
G= \Big[\mbox{ } I_{k-t} \mbox{ } \mid  \mbox{ } \textbf{P}_{k-t, t} \mbox{ }
\Big],
\end{eqnarray}
where $I_{k-t}$ is the identity matrix and $\textbf{P}$ is the sub-matrix that defines the redundant data
$\sum_{i=1}^{k-t} p_{ij}$ to be sent to a set of sources for the purpose
of protection from
multiple link failures, $1\leq j \leq t$. The matrix $G$ is defined explicitly using Maximum Distance Separable (MDS) optimal codes such as Reed-Solomon (RS) codes~\cite{huffman03,macwilliams77}. Based on the above matrix, every source $s_i$ sends
its own message $x_i$ to the receiver $r_i$ via the link $L_i$. In
addition, $t$ edge disjoint paths out of the $k$ edge disjoint paths will carry encoded data.

\begin{definition}\label{defn:mfailuresCode}
An $[k,k-t]_q$ \emph{network protection code} (NPC) is a k-t dimensional subspace of
the space $\F_q^{k}$ that is able to  recover from $t$  edge disjoint path failures. The code protects $k$ working paths and is defined by the matrix $G$ described in Eq.~\ref{eq:Gmultiple}.
\end{definition}

We say that a Network Protection Code (NPC) is feasible/valid on a graph $G$ if  the encoding and decoding operations
can be achieved over the binary field $\F_2$ or a finite field with $q$ elements $\F_q$~\cite{aly08preprint1}. Also, we ensure  that the set of
senders (receivers) are connected with each other.  We define the feasibility conditions of NPC, and we will look for graphs that satisfy these conditions

\smallskip
\noindent
\begin{definition}[NPC Feasibility (validity)]\label{def:NPCfeasible}
Let $S$ and $R$ be sets of source(s) and receiver(s) in a graph $G$, as shown in Fig.~\ref{fig:npaths}.
We say that the network protection code (NPC) is feasible (valid) for $k$ edge disjoint connections (paths) from $s_i$ in $S$ to $r_i$ in $R$, for $i=1, 2,..,k$,
if
\begin{compactenum}[(i)]
\item between any two sources $s_i$ and $s_j$ in $S$, there is a walk (path)
    $s_i\rightarrow s_j$. This means that the nodes in $S$ share a tree;
\item between any two receivers $r_i$ and $r_j$ in $R$, there is a walk (path)
    $r_i\rightarrow r_j$. This means that the nodes in $R$ share a tree;
\item there are k edge disjoint  paths from $S$ to $R$, and the pairs $\langle s_i,r_i\rangle$ are different edge disjoint paths for all $1\leq i \leq k$.
\end{compactenum}
Therefore, we say the graph $G$ is valid for NPC deployment.
\end{definition}
By Definition~\ref{def:NPCfeasible}, there are $k$ edge disjoint paths in the graph from a set of
$k$ senders to a set of $k$ receivers. This also includes the case in which
a single source sends $k$ different messages through $k$ edge disjoint paths to
$k$ receivers, and vice versa. The feasibility of NPC guarantees
that the encoding operations at the senders and decoding operations at the
receivers can be achieved precisely.

\subsection{Problem Statement}

The max edge-disjoint paths (EDP) problem can be defined as follows.  Let
$G=(V,E)$ be  an undirected graph represented by a set of nodes (network
switches, routers, hosts, etc.), and a set of edges (network links, hops,
single connection, etc.).  Assume all edges have the same unit distance,
and they are alike regarding the type of connection that they represent. A
path from a source node  $u$ to a destination node $v$ in $V$ is
represented by a set of edges in $E$. Put differently,
\begin{eqnarray}
\langle u,v\rangle=\{(u,w_1),(w_1,w_2),\ldots,(w_j,v) \mid \nonumber \\ u,v, w_i \in V, ~~~ (w_i,w_{i+1}) \in E \}.
\end{eqnarray}

\noindent {\bf Problem 1.} Given $k$ senders and $k$ receivers.  We can
define  a commodity problem, aka, $k$ edge disjoint paths as follows.
Given a network with a set of nodes $V$, and a set of links $E$, provision
the edge disjoint paths to guarantee the encoding and decoding operations of
NPC.

\begin{compactenum}[i)]
\item
Provision the $k$ edge disjoint paths in $G$. Let
\begin{eqnarray}
L&=&\{L_1,L_2,\ldots,L_k\} \nonumber \\ &=& \{\langle s_j,r_j\rangle \mid~ \forall ~~j=1,\ldots,k, ~ s_i \neq r_i \in V\}
\end{eqnarray} be  a set of commodities.
\item The set of sources $S=\{s_1,\ldots,s_k\}$ are connected with each
    other, and the set of receivers $R=\{r_1,\ldots,r_k\}$ are also
    connected with each other as shown in Definition~\ref{def:NPCfeasible}.
\end{compactenum}
The set $L$  is realizable in $G$ if there exists
mutually edge-disjoint paths from $s_i$ to $r_i$ for all
$i=1,\ldots,k$. Finding the set $L$ in a given arbitrary graph $G$ is an
NP-complete problem as it is similar to edge-disjoint Menger's Problem
and unsplittable flow~\cite{blesa04}.

\noindent \textbf{Problem 2.} Given positive integers $n$ and $k$, and  an NPC with $k$ working
paths, find a $k$-connected n-vertex graph $G$ having the smallest possible
number of edges. This graph by construction must have $k$ edge disjoint paths
and represents a network which satisfies NPC. This problem will be addressed in Section~\ref{sec:graphconstruction}.

\section{NPC and Optimal Graph Construction with Minimum Edges}\label{sec:NPC-minimumedges}
One might ask what  the minimum number of edges on a graph is, in which Network Protection Codes (NPC) is feasible/valid  as stated in Definition~\ref{def:NPCfeasible}. We will answer this question in two cases: (i) the set of sources and receivers are predetermined (preselected from the network nodes), (ii) the sources and receivers are chosen arbitrarily.

\subsection{Single Source and Multiple Receivers}
We  consider the case of a single source and multiple receivers in an arbitrary graph with $n$ nodes.
\begin{lemma}
Let $G$ be a connected graph representing a network with $n$ total nodes, among them a single source node  and $k$ receiver nodes. Assume a NPC from the source node to the multiple receivers is applied. Then, the minimum number of edges required to construct the graph $G$ is given by
\begin{eqnarray}
n+k-2.
\end{eqnarray}
\end{lemma}
\begin{proof}
The graph $G$ contains a single source, $k$ receiver nodes, and $n-k-1$ relay nodes (nodes that are not sources or receivers). To apply NPC, we must have $k$ edge disjoint paths from the source to the $k$ receivers. Also, all receivers must be connected by a tree with a minimum of $k-1$ edges. The remaining $n-k-1$ relay nodes in G can be connected with at least $n-k-1$ edges. Therefore, the minimum number of nodes required to construct the graph G is given by
\begin{eqnarray}
k+(k-1)+(n-k-1).
\end{eqnarray}
\end{proof}

\subsection{Multiple Sources and Multiple Receivers}

\begin{lemma}
Let $G$ be a connected graph with $n$ nodes, and predetermined $k$ sources and $k$ receivers. Then, the minimum number of edges required for predetermined $k$ edge-disjoint paths for a feasible NPC solution on G is given by
\begin{eqnarray}
E_{min}=n+k-2
\end{eqnarray}
\end{lemma}
\begin{proof}
We proceed the proof by constructing the graph $G$ with a total number of nodes $n$ and $k$ sources (receivers).
\begin{compactenum}[i)] \item There are $k$ sources that require a $k-1$ edges represented by a tree.  There are $k$ receivers that require a $k-1$ edges represented by a tree. \item Assume every source node $s_i$ is connected with a receiver node $r_i$ has  an $l_i$ nodes in between for all $1\leq i \leq k$. Therefore,  there are $l_i+1$ edges in every edge-disjoint path, and hence the  number of edges from the sources to receivers is given by $k+\sum_{i=1}^k l_i$.

\item Assume an arbitrary node $u$ exists in the graph $G$, then this node can be connected to a source (receiver) node or to another relay node. In either case, one edge is required to connect this node $u$ to at least one node in $G$. Hence, the number of edges required for all other relay nodes is given by $n-(2k+\sum_{i=1}^k l_i)$. \item Therefore, the total number of edges is given by \begin{eqnarray}E_{min}\!\!&=&\!\!2(k-1)+\big(k+\sum_{i=1}^kl_i\big)+\big(n-(k+\sum_{i=1}^k l_i)\big) \nonumber \\&=&n+k-2\end{eqnarray} \end{compactenum}
\end{proof}

In the previous Lemma, we assume that the $k$ sources and $k$ receivers can be predetermined to minimize the number of edges on $G$. In  Lemma~\ref{lem:minedgesK}, we assume that the sources and receivers can be chosen arbitrarily among the $n$ nodes of $G$.
\begin{lemma}\label{lem:minedgesK}
Let $G=(V,E)$ be a connected graph with  $n$ nodes, and arbitrarily chosen  $k$ sources and $k$ receivers. Then, the minimum number of edges required for any $k$ edge-disjoint paths for a feasible NPC solution on $G$ is given by
\begin{eqnarray}
E_{min}=\lceil n (n-k+1)/2 \rceil
\end{eqnarray}
\end{lemma}
\begin{proof}
In general, assume there are $k$ connection paths, and the source and destination nodes do not share direct connections. In this case, every source node $s_i$ in $V$ must be connected to some relay nodes which are not receivers (destinations). Therefore, every source node must have a node degree of $(n-k+1)$. This agrement is also valid for any receiver node $r_i$ in $V$. If we consider all $n=|V|$ nodes in the graph $G$, hence the total minimum number of edges $E$ must be:
\begin{equation}
\lceil n(n-k+1)/2  \rceil
\end{equation}
The ceiling value comes from the fact that both $n$ and $(n-k+1)$ should not be odd.
\end{proof}

\section{Edge Disjoint Paths in $k$-connected and Regular Graphs}\label{sec:k-connectedgraph}
  In this section, we look for certain graphs where NPC is feasible according to Definition~\ref{def:NPCfeasible}. We will consider two cases: single source single receiver and  single source multiple receivers.   We derive bounds on the cases of $\kappa_e(G)$ and
$\kappa_v(G)$ connectivity in a $k$-connected graph $G$, the definitions are states in the Appendix Section.

\subsection{Single Source to Single   and Multiple Receivers}

Whitney~\cite[Theorem 5.3.6]{gross99} showed that the $k$-connected graph must have $k$ edge disjoint paths
between any two pair of nodes as shown in the following
Theorem.

\begin{theorem}[Whitney 1932]\label{th:whitneytheorem}
A nontrivial graph $G$ is $k$-connected if and only if for each pair $u,
v$ of vertices, there are at least $k$ internally edge disjoint $\langle u,v \rangle$ paths in
$G$.
\end{theorem}


Theorem~\ref{th:whitneytheorem} establishes conditions for $k$ edge disjoint paths in a $k$-connected graph $G$. In order to make NPC feasible in a k-connected graph $G$, we require more two conditions according to Definition~\ref{def:NPCfeasible}: all receivers are connected with each others, as well as all source(s) are connected with each other.

\begin{lemma}\label{lem:kconnectedNPCfeasible}
Let $G$ be a non-trivial graph with a source node $s$ and a receiver node
$r$. Then, the NPC has a feasible solution with at least $k$ edge disjoint paths \textbf{if and only if}
$G$ is a $k$-edge connected graph.
\end{lemma}
\begin{proof}
First, we know that if $G$ is $k$-edge connected, then for each pair $s$ and $r$ of
vertices, the degree of each node must be at least $k$. If not, then
removing any number of edges less than $k$ will disconnect the graph, and
this contradicts the k-edge connectivity assumption. Each node connected with
$s$ will be a starting  path to $r$ or to another node in the graph.
Consequently, every node must have a degree of at least $k$, and must have
a path to $r$. Therefore, there are at least $k$ internally edge disjoint $s-r$
paths in $G$. Hence, NPC is  feasible by considering at least $k$ edge disjoint paths.

Assume that NPC has a feasible(valid) solution for $k^\prime \geq k$, then there must exist $k^\prime \geq k$ edge
disjoint paths in $G$. Then, for each pair of vertices $s$ and $r$, there
are at least $k$ internally edge disjoint paths such that  $\kappa_s(G), \kappa_r(G) \geq k$, for
each $s$ and $r$ non-adjacent nodes. Therefore, the graph $G$  is
$k$-connected.

\end{proof}

Let $s$ be a source in the network model $\N$ that sends $k$ different data stream to $k$ receivers denoted by $R$. We need to infer conditions for all $R$ receivers to be connected with each other, and there are $k$ edge disjoint paths from $s$ to $R$. In this case, NPC will be feasible in the abstract graph $G$ representing the network $\N$.

\begin{theorem}\label{lem:1smultipleR}
Let $G$ be a $k$-edge connected graph with a Hamiltonian cycle, and $s, r_1,r_2,\ldots,r_k$ be any $k+1$
distinct nodes in $G$. Then,
\begin{compactenum}[i)]
\item there is a path $L_i$ from $s$ to $r_i$, for $i=1,\ldots,k$, such
    that the collection $\{L_1,L_2,\ldots,L_k\}$ are internally edge disjoint paths,
\item the nodes in the set $R$ are connected with each other.
\end{compactenum}
Therefore, NPC is feasible in the graph $G$.
\end{theorem}
The second condition ensures that there exists a tree in the graph $G$ which connects all nodes in $R$ without repeating edges from the $k$ edge disjoint paths from $s$ to $R$.

\subsection{Regular Graphs and NPC Feasibility}
We look to establish conditions on regular graphs where it is possible to apply NPC according to Definition~\ref{def:NPCfeasible}.
\begin{theorem}
Let $G$ be a regular graph with minimum degree $k$. Then $G$ has a $k$ edge
disjoint paths if and only if the min-cut separating a source from a sink
is of at least $k$.
\end{theorem}
As shown in Fig.~\ref{fig:regularGNPC}, the degree of each node is three and the min. cut separating the source node $s$ from the receivers $R$ is also three. However, we have the following negative result about NPC feasibility in regular graphs.
\begin{lemma}
There are regular graphs with node degree k, in which NPC is not feasible.
\end{lemma}
\begin{proof}
A certain example to prove this lemma would be a graph of $10$ nodes, each of degree three, separated into two equally components connected with an edge, see Fig.~\ref{fig:regularGNPC}.
\end{proof}

\begin{figure}
\begin{center}  \includegraphics[scale=0.4]{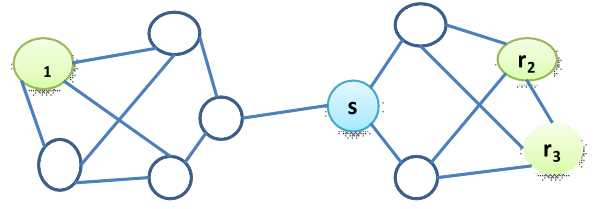}
  \caption{An example of a regular graph of node degree three, in which NPC is not feasible/valid from a source node $s$ to receiver ndoes $r_1,r_2$ and $r_3$. The receivers do not share a tree after removing the node $s$.}\label{fig:regularGNPC}
\end{center}\end{figure}

\section{Graph Construction}\label{sec:graphconstruction}
We will construct graphs with a minimum number of edges for given certain number of vertices $n$ and edge disjoint paths (connections) $k$, in which NPC can be deployed.  Let $h_k(n)$ denote the minimum number of edges that a $k$-connected graph on $n$ vertices must have.
It is shown by F. Harary in 1962 that one can construct a $k$-connected graph $H_{k,n}=(\textbf{V},E)$ on $|\textbf{V}|=n$ vertices that has exactly
$|E|=\lceil\frac{kn}{2}\rceil$ edges for $k \geq 2$. The construction begins
with an $n$-cycle graph, whose vertices $\textbf{V}$ are consecutively numbered
$v_0,v_1,v_2,\ldots, v_{n-1}$ clockwise.
The proof of the following Lemma is shown in~\cite[Proposition
5.2.5.]{gross99}.
\begin{lemma}
Let $G$ be a $k$-connected graph with $n$ nodes. Then, the number of edges in
$G$ is at least $\lceil \frac{kn}{2}\rceil$. That is, $h_k(n) \geq \lceil
\frac{kn}{2} \rceil$.
\end{lemma}


\begin{algorithm}[t!]
\SetLine%
\KwIn{Two positive integers $k$ and $n$, number of connections and vertices, such that $k <n$.}%
\KwOut{An optimal $H_{k,n}$ Harary graph for NPC. \\}%

\textbf{Begin:}
Scatter the $n$ isolated nodes \;
Let $r=\lfloor k/2\rfloor$. \;

The construction of $H_{2r,n}$ Harary graph.\;
 \ForEach{i=0 to n-2}
  { \ForEach {j=i+1 to n-1}
    {  \lIf{ $j-i \leq r $ OR $n+i-j \leq r $}\\
    {Create an edge between vertices $v_i$ and $v_j$.}

    }
   }
\uIf{$k$ is even}
{
    Return graph $H$ for NPC.
}
   \Else
   {
    \uIf{$n$ is even}
    {
\ForEach{i=0 to $\frac{n}{2}-1$} {Create an edge between vertex $v_i$ and
vertices $v_{i+\frac{n}{2}}$}
    }
    \Else
    {Create an edge from vertex $v_0$ to vertex $v_\frac{n-1}{2}$\;
     Create an edge from vertex $v_0$ to vertex $v_\frac{n+1}{2}$\;
     \ForEach{i=1 to $\frac{n-3}{2}$}
     {Create an edge between vertex $v_i$ and vertex $v_{i+\frac{n+1}{2}}$}
    }
Return graph H for NPC
   }

\mbox{}\\ \caption{Construction of an optimal graph for NPC with given n
vertices, k connections and minimum number of edges.}
\label{alg:hararyNPC}
\end{algorithm}

From Algorithm~\ref{alg:hararyNPC}, one can ensure that there are $k$ edge
disjoint paths between any two nodes (one is sender and one is receiver).
In addition, there are $k$ edge disjoint paths from any node, which acts as a
source, and $k$ different nodes, which act as receivers. All nodes are
connected together with a loop. Therefore, NPC can be deployed to such
graphs.


Due to the fact that Harary's graph is $k$-connected~\cite[Theorem 5.2.6.]{gross99}, then using our previous result, we can deduce that NPC is feasible for such graphs.
Harary's graphs are optimal for the NPC construction in the sense that they are
$k$-edge connected graphs with the fewest possible number of edges.
%
%
\section{Conclusion}\label{sec:conclusion}
In this paper, we proposed  graph topologies for  network protection using network coding.  We derived bounds on the minimum number of edges and showed a method to  construct optimal network graphs. \nix{ based on Harary graph constructions. Our future work will include practical aspects of the proposed models.}

\emph{
\noindent Network protection is much easier than human protection against failures. S.~A.~A.
}
\bibliographystyle{plain}

\bibliographystyle{ieeetr}

\section*{Appendix:  Essential Definitions}\label{sec:definitions}
\nix{In this Appendix, we  state  some essential definitions.} We
assume that all graphs stated in this paper are undirected (bi-directional edges) unless stated
otherwise. We  define the
edge-connectivity and node-connectivity of a graph $G$ as follows.
%
%
\begin{definition}[Edge-cut and node-cut] Given an undirected connected graph $G$, an edge-cut in a graph $G$ is a set of edges such that its removal
    disconnects the graph. A node-cut in  $G$ is a set of nodes such that its removal disconnects the graph.
\end{definition}


%
\begin{definition}[node-edge-connectivity]
The edge-connectivity of a connected graph $G$, denoted $\kappa_e(G)$ is
the size of a smallest edge-cut. Also, the node-connectivity  of a
connected graph $G$, denoted $\kappa_v(G)$ is the minimum number of
vertices whose removal can either disconnect the graph $G$ or reduce it to
a one-node graph.
\end{definition}

The connectivity measures $\kappa_v(G)$ and $\kappa_e(G)$ are used in a quantified model of network survivability, which is the capacity of a network to retain connections among its nodes after some edges or nodes are removed.


\begin{definition}[k-connected graph]
a graph $G$ is k-node connected if $G$ is connected and $\kappa_v(G) \geq k$.
Also, a graph $G$ is $k$-edge connected if $G$ is connected and every
edge-cut has at least $k$ edges, $\kappa_e(G) \geq k$.
\end{definition}

We define two internal connections (paths) between nodes $u$ and $v$ in a graph $G$ to be internally edge
disjoint if they have no edge in common. This is also different from the node disjoint paths. \nix{, in which no two disjoint paths share a common node.} Throughout this paper, a path $\langle u,v \rangle$ from a starting node $u$ to an ending node $v$ is a walk in a graph $G$~\cite{gross99}, i.e., it does not contain the same node or edge twice.
\end{document}